\documentclass[12pt]{article}

\usepackage[titletoc,title]{appendix}
\usepackage[latin1]{inputenc}
\usepackage{amsmath}
\usepackage{amssymb}
\usepackage{booktabs}
\usepackage{graphicx}
\usepackage[margin=1.25in]{geometry}
\usepackage[bottom]{footmisc}
\usepackage{indentfirst}
\usepackage{endnotes}
\usepackage{mathabx}
\usepackage{enumerate}
\usepackage{rotating}
\usepackage{multirow}
\usepackage{booktabs,calc}
\usepackage{threeparttable}
\usepackage{float}
\usepackage{url}
\usepackage{amsmath,amsfonts}
\usepackage{amsthm}
\usepackage[onehalfspacing]{setspace}

\newcommand{\suchthat}{\;\ifnum\currentgrouptype=16 \middle\fi|\;}

\newtheorem{assm}{Assumption}

\newtheorem{prop}{Proposition}
\newtheorem{lemma}{Lemma}
\newtheorem{cor}{Corollary}

\newtheorem{example}{Example}
\newtheorem{defn}{Definition}

\theoremstyle{remark}

\usepackage{chngcntr}
\usepackage{graphicx}
\usepackage{amsmath}
\usepackage{thmtools}
\usepackage{amssymb}
\usepackage{parskip}
\usepackage{sgame}
\usepackage{color}
\usepackage{tikz}
\usetikzlibrary{trees,calc}

\DeclareMathOperator*{\argmax}{\arg\!\max}

\newcommand{\E}{\mathbb{E}}

\newcommand{\eps}{\varepsilon}

\usepackage{natbib}

\begin{document}

\title{\large{Injectivity and the Law of Demand}\thanks{I thank Nail Kashaev, Salvador Navarro, John Rehbeck, and David Rivers for helpful comments.}}

\author{ \small{Roy Allen}  \\
    \small{Department of Economics} \\
    \small{University of Western Ontario} \\
    \small{rallen46@uwo.ca}
}
\date{\small{ \today }} 

\maketitle

\begin{abstract}
Establishing that a demand mapping is injective is core first step for a variety of methodologies. When a version of the law of demand holds, global injectivity can be checked by seeing whether the demand mapping is constant over any line segments. When we add the assumption of differentiability, we obtain necessary and sufficient conditions for injectivity that generalize classical \cite{gale1965jacobian} conditions for quasi-definite Jacobians.

\end{abstract}

\newpage

\section{Introduction}

A variety of recently developed methods require, as a first step, that a demand mapping be injective. Examples include work on endogeneity with market level data (\cite*{berry1994estimating}, \cite*{berry1995automobile}, \cite{berry2009nonparametric}, \cite*{chiappori2009nonparametric}, \cite*{berry2014identification}); simultaneous equations models (\cite*{matzkin2008identification}, \cite*{matzkin2015estimation}, \cite*{berry2018identification}); multidimensional heterogeneity in a consumer setting (\cite{blundell2017individual}); index models (\cite*{ahn2017simple}); and nonparametric analysis in trade (\cite{adao2017nonparametric}).\footnote{\cite*{chesher2017generalized} and \cite{bonnet2017yogurts} take an alternative approach, working with inverse images that may be multivalued.}

The applicability of these methods depends on whether the demand mapping is injective. This paper studies injectivity using a shape restriction that allows complementarity: the law of demand.
\begin{defn}
$Q: U \subseteq \mathbb{R}^K \rightarrow \mathbb{R}^K$ satisfies the \textbf{law of demand} if for each $u, \tilde{u} \in U$,
\[
( Q(u) - Q(\tilde{u}) ) \cdot ( u- \tilde{u} ) \geq 0.
\]
\end{defn}
Many models imply a version of the law of demand, both in the standard consumer problem and outside it. In the standard consumer problem, $u$ is the \textit{negative} of the price vector. Quasilinear preferences imply the law of demand.\footnote{\cite{fosgerau2018demand} provide an injectivity result for a demand system that allows complementarity. They consider quasilinear preferences and so their demand system fits into the setup of this paper.} \cite*{hildenbrand1983law} provides conditions under which the law of demand holds in the aggregate, even if it does not hold at the individual level.\footnote{\cite{hildenbrand1983law} also provides sufficient conditions that ensure a strict law of demand $(Q(u) - Q(\tilde{u})) \cdot (u - \tilde{u}) > 0$ for $u \neq \tilde{u}$, which clearly implies injectivity.} Outside the standard consumer problem, the discrete choice additive random utility model (\cite*{mcfadden1981}) also satisfies the law of demand. In that model one may interpret $u_k$ as the deterministic utility index for alternative $k$ and $Q(u)$ as a vector of choice probabilities.

Directly checking whether a demand mapping is injective is nontrivial. This paper provides necessary and sufficient conditions for injectivity that can simplify this task. The simplest condition states that when demand is continuous and the domain of demand is convex, global injectivity of $Q$ can be checked by checking whether $Q$ is constant over line segments. This implies that global and local injectivity are equivalent.

The main result of this paper is a nondifferentiable counterpart to the classical injectivity results of \cite{gale1965jacobian} for functions with weakly quasi-definite Jacobians.\footnote{To be clear, this paper only overlaps when the Jacobian is weakly quasi-definite, not just a $P$ matrix as in Theorem 4 in \cite{gale1965jacobian}.} When I specialize the main result by assuming the demand mapping $Q$ is differentiable, I establish a generalization of \cite{gale1965jacobian}, Theorem 6w. While \cite{gale1965jacobian} impose invertibility of the Jacobian of $Q$ as a \textit{sufficient} condition for global injectivity, I provide a necessary and sufficient condition for local (and global) injectivity in terms of certain directional derivatives.

\cite*{berry2013connected} have recently shown that demand mappings that satisfy a ``connected substitutes'' property are injective. This connected substitutes condition applies to a number of existing models, including models of market shares based on a discrete choice foundation (e.g. \cite{berry2014identification}), but may not apply when there is complementarity.\footnote{\cite{berry2013connected} show in several examples that certain models with complements may be reparametrized to fit into their setup. See also \cite{brown1998estimation} and \cite{beckert2008heterogeneity} for injectivity results that allow complementarity between goods.} This paper complements their analysis by studying injectivity using a shape restriction that allows complementarity without a reparametrization.

The injectivity results of this paper exploit the fact that when $Q$ is continuous and satisfies the law of demand, the inverse image of any quantity is a convex set. This is a classical result in monotone operator theory.\footnote{See e.g. \cite{rock2009variational}.} To my knowledge, this important property has not been exploited for studying injectivity in the econometrics literature, yet it has several implications that I describe further in the paper. The closest precedent appears to be in the study of uniqueness of general equilibrium, where several conditions are known to yield convexity of equilibria (e.g. \cite*{arrow1958stability}, \cite*{arrow1960some}, and the discussion in \cite*{mas1991uniqueness}).
\section{Characterization of Injectivity}
This section presents the main results, which provide necessary and sufficient conditions for a demand mapping to be injective. We use the following assumption, which allows us to reduce checking global injectivity to checking local conditions.
\begin{assm} \label{assm:main}
$Q: U \subseteq \mathbb{R}^K \rightarrow \mathbb{R}^K$ satisfies the law of demand, is continuous, and $U$ is open and convex.
\end{assm}
Recall that a set $U$ is convex if for $u,\tilde{u} \in U$ and any scalar $\alpha \in [0,1]$, it follows that $\alpha u + (1 - \alpha) \tilde{u} \in U$. An important implication of this assumption is that inverse images $Q^{-1}(u)$ are convex, which I formalize below. As discussed in the Introduction, this is a classical result in monotone operator theory.
\begin{lemma} \label{lem:invconv}
Let Assumption~\ref{assm:main} hold. Then for each $y \in \mathbb{R}^K$,
\[
Q^{-1}(y) = \{ u \in U \mid Q(u) = y \}
\]
is convex.
\end{lemma}
\begin{proof}
If the domain is $U = \mathbb{R}^K$, this is a textbook result, e.g.
\cite{rock2009variational}, p. 536. When $U \neq \mathbb{R}^K$, this is covered by \cite*{kassay2009convexity}, Theorem 3.5.
\end{proof}
When $K > 1$, continuity cannot be dropped without alternative structure, as the following example illustrates.
\begin{example}
Let $K = 2$, $U = \mathbb{R}^2$, and $A = \{ u \in \mathbb{R}^2 \mid u_1 + u_2 > 0 \text{ or } u_1 = u_2 = 0 \}$. Let $Q(u) = ( 1\{ u \in A\}, 1\{ u \in A \})$, where $1\{ u \in A \}$ is an indicator function for whether $u \in A$. To show the law of demand is satisfied, note that if $u,\tilde{u}$ are either both in $A$ or both in its complement $A^c$, $Q$ does not vary and clearly satisfies the law of demand. Consider then $u \in A, \tilde{u} \in A^c$. Then we have
\[
(Q(u) - Q(\tilde{u}))'(u - \tilde{u}) = (u_1 - \tilde{u}_1) + (u_2 - \tilde{u}_2) \geq 0,
\]
Nonetheless, $Q^{-1}(0,0) = A^c$ is not convex, since both points $(-1,1), (1,-1)$ are in $A^c$, yet their convex combination $(0,0)$ is not.
\end{example}

Using Lemma~\ref{lem:invconv}, we obtain the following list of conditions that are equivalent to global injectivity of $Q$.
\begin{prop} \label{cor:three}
Let Assumption~\ref{assm:main} hold. Then the following are equivalent:
\begin{enumerate}[i.]
    \item $Q$ is injective, i.e. for each $y \in \mathbb{R}^K$ there is at most one $u \in U$ such that $Q(u) = y$.
    \item $Q$ is locally injective, i.e. for each $u \in U$ there is a neighborhood $H \subseteq \mathbb{R}^K$ of $u$ such that the restriction of $Q$ to $H$ is injective.
    \item The only line segments in $U$ along which $Q$ is constant are singleton points.
\end{enumerate}
\end{prop}
\begin{proof}
Clearly, (i) $\implies$ (ii) $\implies$ (iii). That (iii) $\implies$ (i) follows from Lemma~\ref{lem:invconv} and the definition of convexity. Indeed, if $Q(u) = Q(\tilde{u})$, then the set $Q^{-1}(Q(u))$ is convex and must contain $u$ and $\tilde{u}$. In particular, $Q$ is constant on any line segment joining $u$ and $\tilde{u}$. By assumption, this is only possible if $u = \tilde{u}$.
\end{proof}
Local injectivity always implies condition (iii), but in general the reverse is not true. An example of a multivariate function that satisfies (iii) but is not locally injective is $Q(u_1,u_2) = (u_1^3 - u_2, u_1^3 - u_2)$. The equivalence of (i) and (iii) is the most powerful part of Proposition~\ref{cor:three}, since part (iii) is often easy to check. In addition, we can leverage this equivalence to relax the domain restrictions. We formalize this as follows.
\begin{cor} \label{cor:convex}
Let Assumption~\ref{assm:main} hold. Let $Q_H : H \subseteq U \rightarrow \mathbb{R}^K$ denote the restriction of $Q$ to the set $H$. If $H$ is open or convex, any of the injectivity conditions of Proposition~\ref{cor:three} are equivalent when applied to the function $Q_H$.
\end{cor}
\begin{proof}
It is clear that (i) $\implies$ (ii) $\implies$ (iii). It remains to show that (iii) $\implies$ (i).

Let $Q_H(u) = Q_H(\tilde{u})$ for $u, \tilde{u} \in H$. Let $T \subseteq U$ be the line segment from $u$ to $\tilde{u}$. From Lemma~\ref{lem:invconv}, we conclude that $Q^{-1}(Q_H(u))$ contains $T$, and hence $Q^{-1}_H (Q_H(u))$ contains $T \cap H$. In particular, $Q_H$ is constant over $T \cap H$.

Suppose for the purpose of contradiction that $u \neq \tilde{u}$. If $H$ is either open or convex, the set $T \cap H$ contains two distinct line segments that are not points, beginning at $u$ and $\tilde{u}$, respectively. Since we assumed (iii) holds, we reach a contradiction since $Q_H$ is constant over $T \cap H$, and thus is constant over these line segments.
\end{proof}

Note that this proposition assumes $Q$ satisfies the law of demand over the entire set $U$. This assumption may be satisfied by appealing to economic theory. It allows one to show injectivity for the restriction $Q_H$ for a set $H$ that is either open \textit{or} convex. The primary reason one may be interested in such restrictions is that one may only have information on $Q$ over a certain region of utility indices (such as $H$). 

The proof of Corollary~\ref{cor:convex} establishes the following additional result.
\begin{cor} \label{cor:localone}
Let Assumption~\ref{assm:main} hold. Let $Q_H : H \subseteq U \rightarrow \mathbb{R}^K$ denote the restriction of $Q$ to the set $H$. If $H$ is open or convex, the following are equivalent for arbitrary $u \in H$:
\begin{enumerate}[i.]
    \item For any $u, \tilde{u} \in H$ with $u \neq \tilde{u}$, $Q_H(u) \neq Q_H(\tilde{u})$, i.e. the inverse image $Q_H^{-1}(Q_H(u))$ is a singleton.
    \item $Q_H$ is locally injective at $u$, i.e. there is a neighborhood $N \subseteq H$ of $u$ such that the restriction of $Q_H$ to $N$ is injective.
    \item The only line segment in $H$ that contains $u$ and and over which $Q_H$ is constant is the point $u$.
\end{enumerate}
\end{cor}
The equivalence of (i) and (ii) shows that to check whether an inverse image is a singleton, it is enough to check features of the mapping $Q_H$ that are local to a \textit{single} $u$. Importantly, this differs from Proposition~\ref{cor:three} and Corollary~\ref{cor:convex}, which instead study how local injectivity holding \textit{for each} $u$ implies global injectivity. Thus, Corollary~\ref{cor:localone} further highlights the sense in which injectivity can be reduced to a local condition. It also conceptually differs from classical papers such as \cite{gale1965jacobian}, which focus on when a local condition holding \textit{everywhere} implies global injectivity.

Finally, to see that these equivalences do not hold in general, consider $H = U = \mathbb{R}$ and $Q(u) = u^2$, which is continuous yet violates the law of demand. Then $Q_H$ is locally injective at $u = 1$, but $Q_H^{-1}(Q_H(1)) = \{ -1, 1 \}$. 
\section{Relationship to \cite*{gale1965jacobian}}
In this section I describe how the local-to-global injectivity result of Proposition~\ref{cor:three} may be seen as a nondifferentiable version of a classical result due to \cite*{gale1965jacobian}. In drawing this relationship, I present a new result complementing their results for weakly quasi-definite Jacobians, which drops the requirement that the function have a convex domain.

I now add the assumption that $Q$ is differentiable.
\begin{assm} \label{assm:diff}
The function $Q: U \subseteq \mathbb{R}^K \rightarrow \mathbb{R}^K$ is differentiable, where $U$ is an open, convex set.
\end{assm}

To relate to \cite*{gale1965jacobian}, I introduce some definitions.
\begin{defn}
A $K \times K$ matrix $B$ is positive semi-definite if $\lambda' B \lambda \geq 0$ for every $\lambda \in \mathbb{R}^K$. If $\lambda' B \lambda > 0$ for every nonzero $\lambda$, then $B$ is positive definite.
\end{defn}
\begin{defn}
A $K \times $K matrix $B$ is weakly quasi-definite if $(B + B')/2$ is positive semi-definite. If $(B + B')/2$ is positive definite, then $B$ is quasi-definite.
\end{defn}
The following result connects the law of demand and quasi-definiteness of Jacobians.
\begin{lemma} \label{lem:definitelaw}
Let Assumption~\ref{assm:diff} hold. The function $Q$ satisfies the law of demand if and only if its Jacobian is everywhere weakly quasi-definite.
\end{lemma}
\begin{proof}
See e.g. \cite*{parthasarathy2006global}, p. 92.
\end{proof}
With this lemma and the previous results, we obtain a generalization of \cite*{gale1965jacobian}, Theorem 6w.
\begin{prop} \label{prop:galecomp}
Let Assumption~\ref{assm:diff} hold and suppose $Q$ satisfies the law of demand. Let $H \subseteq U$ be an open set and let $Q_H$ denote the restriction of $Q$ to $H$. Then $Q_H$ is injective if its Jacobian is everywhere invertible.
\end{prop}
\begin{proof}
From Corollary~\ref{cor:convex} we see it is enough to establish local injectivity of $Q_H$. Since the Jacobian of $Q_H$ is everywhere invertible, $Q_H$ is locally injective by Proposition~\ref{prop:diff} because its directional derivatives are never zero.\footnote{Note that if we had assumed $Q$ is \textit{continuously} differentiable, we could apply the classical inverse function theorem to establish local injectivity. We have not assumed the derivative is continuous, and hence we use an alternative technique.}
\end{proof}
\cite*{gale1965jacobian} prove this result when $H = U$, i.e. over convex domains.\footnote{Theorem 6 in that paper requires instead that $U$ be convex and drops openness, but instead also requires that the Jacobian be positive quasi-definite (not just weakly quasi-definite, which is implied by the law of demand).} The link between the law of demand and results in \cite*{gale1965jacobian} has previously been noted (\cite*{kassay2009convexity}, \cite*{laszlo2016injectivity}). Proposition~\ref{prop:galecomp} is new to my knowledge, and handles the important case of checking injectivity over a nonconvex set.

\section{Relaxing the Jacobian Condition} \label{sec:inverse}
The previous section uses the invertibility of the Jacobian of $Q$ as a sufficient condition for local injectivity of $Q$. Invertibility of the Jacobian is not necessary for local injectivity, as illustrated for $K = 1$ by $Q(u) = u^3$, since the derivative is $0$ at $0$. 

Obtaining \textit{necessary and sufficient} conditions for local injectivity in terms of derivatives is nontrivial in general. For $Q$ satisfying the law of demand, by Proposition~\ref{cor:three}(iii), however, checking global or local injectivity is equivalent to checking whether $Q$ is constant over any line segment that is not a point. One may write this condition in terms of certain directional derivatives. To that end, define the directional derivative of $Q$ at $u$ in direction $v$, denoted $Q'(u,v)$, by
\[
\left| \lim_{\lambda \downarrow 0} \frac{Q(u + \lambda v) - Q(u)}{\lambda } - Q'(u,v) \right| = 0
\]
whenever this limit exists.
\begin{prop} \label{prop:diff}
Let Assumption~\ref{assm:diff} hold and suppose $Q$ satisfies the law of demand. The following are equivalent for arbitrary $u \in U$:
\begin{enumerate}[i.]
    \item For any $\tilde{u} \in U$ with $u \neq \tilde{u}$, $Q(u) \neq Q(\tilde{u})$.
    \item $Q$ is locally injective at $u$.
    \item The only line segment in $U$ that contains $u$ and and over which $Q$ is constant is the point $u$.
    \item There are no nonzero vectors $v \in \mathbb{R}^K$ such that for all $\lambda \in [0,1]$ satisfying $u + \lambda v \in U$, $Q'(u + \lambda v, v)$ is the zero vector.
\end{enumerate}
\end{prop}
\begin{proof}
Since $Q$ is differentiable, it is continuous. The equivalence between (i)-(iii) follows from Corollary~\ref{cor:localone}. Equivalence between (iii) and (iv) follows from the mean value theorem.
\end{proof}
Note that if the Jacobian of $Q$ is invertible at $u \neq 0$, then $Q'(u,v)$ cannot be the zero vector for any nonzero $v$, because directional derivatives satisfy $J(u) v = Q'(u,v)$, where $J(u)$ is the Jacobian of $Q$ at $u$. Thus if $Q$ has an everywhere invertible Jacobian, condition (iii) is satisfied for each $u \in U$. More generally, suppose that failures of invertibility of the Jacobian of $Q$ only occur on an isolated set of points. Then clearly, condition (iii) is satisfied for each $u \in U$.

Recall from Proposition~\ref{cor:three}, global injectivity of $Q$ is equivalent to local injectivity for each $u \in U$. Thus, global injectivity of $Q$ is also equivalent to conditions (ii) or (iii) of Proposition~\ref{prop:diff} holding for each $u \in U$.

By combining Lemma~\ref{lem:definitelaw} and Proposition~\ref{prop:diff}, we obtain a  generalization of \cite*{gale1965jacobian}, Theorem 6. This generalization drops the assumption that the Jacobian of $Q$ is everywhere invertible.
\begin{cor} \label{cor:gale}
Let Assumption~\ref{assm:diff} hold and assume the Jacobian of $Q$ is everywhere weakly quasi-definite. The following are equivalent:
\begin{enumerate}[i.]
    \item $Q$ is injective.
    \item For each $u \in U$, any of the equivalent conditions in Proposition~\ref{prop:diff} holds.
\end{enumerate}
\end{cor}

\section{Complements, Substitutes, and the Law of Demand}

\cite*{berry2013connected} have recently shown that a ``connected substitutes'' condition implies global injectivity. The present paper shows that a version of the law of demand, which allows complementarity, also suffices. This approach is not nested in and does not nest that of \cite*{berry2013connected}.

The setup of \cite*{berry2013connected} imposes the following properties on the demand mapping $Q : U \subseteq \mathbb{R}^K \rightarrow \mathbb{R}^K$:
\begin{enumerate}[i.]
    \item (Strict Own-Good Monotonicity) Let $k$ be arbitrary. For each $u, \tilde{u} \in U$ such that $u_k > \tilde{u}_k$ and $u_j = \tilde{u}_j$ for $j \neq k$, it follows that
    \[
    Q_k(u) > Q_k(\tilde{u}).
    \]
    \item (Weak Substitutability) Let $k$ be arbitrary. For each $u, \tilde{u} \in U$ such that $u_k > \tilde{u}_k$ and $u_j = \tilde{u}_j$ for $j \neq k$, it follows that for all $\ell \neq k$,
    \[
    Q_{\ell}(u) \leq Q_{\ell}(\tilde{u}).
    \]
\end{enumerate}
Condition (i) states demand increases in its own utility shifter. Condition (ii) states that if a utility shifter for good $k$ increases, then all other demands weakly decrease. Condition (i), except with a weak inequality, follows whenever $Q$ satisfies the law of demand, but Condition (ii) does not.

\cite*{berry2013connected} impose a ``connected substitutes'' condition, which directly assumes weak substitutability and implies strict own-good monotonicity (see Remark 1 in \cite*{berry2013connected}). For brevity, I omit a formal statement, and instead describe a key implication of their assumption: for arbitrary $u, \tilde{u} \in U$,
\[
Q(u) \geq Q(\tilde{u}) \implies u \geq \tilde{u},
\]
where $\geq$ denotes the usual partial order in $\mathbb{R}^K$, i.e. $u \geq \tilde{u}$ if and only if $u_k \geq \tilde{u}_k$ for each $k$. This shape restriction is called \textbf{inverse isotonicity}, and clearly implies that $Q$ is injective.\footnote{Appendix~\ref{app:pfunction} describes a characterization of inverse isotonicity in this setting using results in \cite{more1973p}.}  This property is essential for establishing injectivity using the approach of \cite*{berry2013connected}.

As discussed previously, the methods of this paper are distinct from those of \cite*{berry2013connected}. I provide two examples showing the distinction between inverse isotonicity and the law of demand. First, I show that the law of demand does not imply inverse isotonicity.
\begin{example} \label{ex:1}
Consider a linear demand system $Q(u) = Au$, where
\[
A = 
\begin{bmatrix}
    2 & 1 \\
    1 & 2
\end{bmatrix}.
\]
The matrix $A$ is symmetric and satisfies row-diagonal dominance (e.g. for each row, the diagonal $| 2 |$ exceeds the sum of the off-diagonal $|1|$), which are well-known conditions that ensures the matrix $(A + A')/2$ is positive semi-definite. Thus, the law of demand is established from Lemma~\ref{lem:definitelaw}. This demand mapping violates weak substitutability because the off-diagonals of $A$ are positive. In addition, it violates inverse isotonicity. To see this, consider the two vectors $u = (0,0)$, $\tilde{u} = (2, -1)$. Then $Q(u) = (0,0)$, $Q(\tilde{u}) = (3, 0)$ and so $Q(\tilde{u}) \geq Q(u)$, but we do not have $\tilde{u} \geq u$.
\end{example}

The following example illustrates that inverse isotonicity does not imply the law of demand.

\begin{example} \label{ex:2}
Now consider a demand system
\[
Q(u) = 
\begin{bmatrix}
    20 & -10 \\
    -1 & 2
\end{bmatrix} \begin{bmatrix} u^3_1 \\ u^3_2 \end{bmatrix}.
\]
The function $Q$ satisfies the connected substitutes property of \cite*{berry2013connected}, hence inverse isotonicity. It does not satisfy the law of demand. To see this, consider $u = (0,0)$ and $\tilde{u} = (1,2)$. One obtains $Q(u) = (0,0)$ and $Q(\tilde{u}) = (-60,7)$. Thus,
\[
(Q(\tilde{u}) - Q(u))'(\tilde{u} - u) = (-60,7)'(1,2) = -46 < 0.
\]
In this example, the law of demand fails because the substitution effect outweighs the own-good effect.
\end{example}

To shed further light on the distinction between inverse isotonicity and the law of demand, it is helpful to note that the law of demand is not an ordinal property. This is illustrated in Example~\ref{ex:2} by considering the strictly increasing function $f(v) = v^{1/3}$. Consider the transformed mapping
\[
\tilde{Q}(u) = Q((f(u_1),f(u_2))) =
\begin{bmatrix}
    20 & -10 \\
    -1 & 2
\end{bmatrix} \begin{bmatrix} u_1 \\ u_2 \end{bmatrix}.
\]
The mapping $\tilde{Q}$ satisfies the law of demand even though the original mapping $Q$ in Example~\ref{ex:2} violates the law of demand.\footnote{This follows because the symmetrized matrix of coefficients in $\tilde{Q}$ satisfies diagonal dominance.} Further analysis of the law of demand and a change of variables is covered in Section~\ref{sec:change}.

Finally, it is clear that strict own-good monotonicity, weak substitability, and inverse isotonicity are all ordinal properties in the following sense: they hold for some mapping $Q(u)$ if and only if they hold for $\tilde{Q}(u) = Q(f(u))$ where $f(u) = (f_1(u_1), \ldots, f_K(u_K))$, and each $f_k$ is strictly increasing.

\section{Quasilinear Utility}
In the standard consumer problem, quasilinear utility is a well-known class of preferences that implies the law of demand. I provide an injectivity result that exploits additional structure of this model.

Suppose an individual maximizes a utility function of the form
\[
y_0 + C(y_1, \ldots, y_k),
\]
with budget constraint $\sum_{k = 0}^K p_k y_k \leq M$. Suppose $p_0$ does not vary and is normalized to $1$. Under local nonsatiation and allowing negative quantities of $y_0$ (or sufficiently high income), the maximization problem is equivalent to choosing quantities to maximize
\[
-\sum_{k = 1}^K p_k y_k + C(y_1, \ldots, y_k),
\]
where $y_0 = M -\sum_{k = 1}^K p_k y_k$ has been substituted out. Thus, if the maximizer is unique we have
\[
Q(u) = \argmax_{y \in \mathbb{R}^K} \sum_{k = 1}^K u_k y_k + C(y),
\]
where $u_k = -p_k$; more generally, we may take $Q(u)$ to be an element of the argmax correspondence.

To see that $Q$ satisfies the law of demand, consider the necessary condition for maximization
\[
\sum_{k = 1}^K u_k Q_k(u) + C(Q(u)) \geq \sum_{k = 1}^K u_k Q_k(\tilde{u}) + C(Q(\tilde{u})).\footnote{If there are multiple maximizers, this inequality holds for any maximizers. In particular, a law of demand holds for arbitrary selectors from the argmax correspondence.}
\]
An analogous inequality holds with $u$ and $\tilde{u}$ reversed. Summing up the two analogous inequalities and rearranging establishes that $Q$ satisfies the law of demand.

Quasilinear models imply (weakly) more than the law of demand (\cite*{brown2007nonparametric}). The additional structure of quasilinear models yields the following result. Note that the domain $U$ need not be convex.
\begin{lemma}[cf. \cite*{rock70}, Theorems 23.5 and 25.1] \label{lem:injdif}
Let $Q: U \subseteq \mathbb{R}^K \rightarrow \mathbb{R}^K$ satisfy
\[
Q(u) \in \argmax_{y \in \mathbb{R}^K} \sum_{k = 1}^K u_k y_k + C(y)
\]
where $U$ is open. If $C : \mathbb{R}^K \rightarrow \mathbb{R} \cup \{ + \infty \}$ is concave, upper semi-continuous,\footnote{A function $f : \mathbb{R}^K \rightarrow \mathbb{R} \cup \{ + \infty \}$ is upper semi-continuous if $\{ y \mid f(y) \geq \alpha \}$ is closed for each $\alpha$.} and finite at some point, then the following are equivalent for arbitrary $u \in U$:
\begin{enumerate}[i.]
    \item $C$ is differentiable at $Q(u)$.
    \item There is no $\tilde{u} \in U$ such that $u \neq \tilde{u}$ and $Q(u) = Q(\tilde{u})$.
\end{enumerate}
\end{lemma}
This result directly follows from \cite*{rock70} and so the proof is omitted. A corollary of this lemma is that if $C$ is everywhere differentiable (and the other conditions are met), then $Q$ is globally invertible. A version of this result has been used in \cite*{idsep}; I include this result for completeness, to illustrate how additional structure allows us to further specialize the results, and because the quasilinear structure is widely used.

\subsection{Relation to \cite*{brown1998estimation}}

A structure that shares a mathematical relationship with quasilinear utility has been studied in \cite*{brown1998estimation}. Suppose now that there is a budget constraint but the demand is not quasilinear. Formally, the consumer solves the problem
\begin{align*}
\max_{(y_0, y) \in \mathbb{R}^{K+1}_+} \sum_{k = 1}^K u_k y_k & + C(y_0, y) \qquad \text{ s.t. } \\
& \sum_{k = 1}^K p_k y_k + y_0 \leq I,
\end{align*}
where $I$ is income. \cite*{brown1998estimation} study this model with the interpretation that $u$ is a vector of exogenous unobservables.\footnote{Their presentation is slightly different since they consider also a term $u_0 y_0$ in the utility and then normalize $u_0 = 1$. With this normalization, this term can be absorbed into $C(y_0,y)$.} \cite{blundell2017individual} study a generalization that is not covered by our setup.

Note that fixing prices and income, this problem differs from the setup of Lemma~\ref{lem:injdif} only because of the budget constraint. However, under local nonsatiation the constraint is satisfied with equality and so the problem reduces to
\[
\max_{y \in \mathbb{R}^{K}_+} \sum_{k = 1}^K u_k y_k + C \left(I - \sum_{k = 1}^K p_k y_k, y \right).
\]
Injectivity of this demand system is then covered by Lemma~\ref{lem:injdif}. In particular, fixing prices and income, differentiability of the mapping $\tilde{C}(y) = C \left(I - \sum_{k = 1}^K p_k y_k, y \right)$ and injectivity are equivalent under certain conditions, as formalized in Lemma~\ref{lem:injdif}.\footnote{A related change of variables argument has appeared in \cite{idsep} to study a discrete choice problem.} Importantly, while this argument can provide sharp conditions relating injectivity and differentiability, it does \textit{not} establish smoothness of the inverse. Indeed, differentiability of $\tilde{C}$ does not rule out multiple maximizers, and so the demand mapping need not even be continuous. \cite{brown1998estimation} provide additional conditions that ensure injectivity \textit{and} smoothness.

\subsection{Discrete Choice and Aggregation}
Many models outside of the consumer problem that have additively separable unobservable heterogeneity also share the structure of the quasilinear utility model. In particular, they imply a version of the law of demand in utility indices that need not involve price. For the additive random utility model, this has been recognized at least since the seminal work of \cite*{mcfadden1981}. Other examples that share this structure are the discrete choice bundles model of \cite*{gentzkow2007valuing}, the matching model of \cite*{fox2018unobserved}, and a model of decisions under uncertainty considered in \cite*{agarwal2018demand}; see \cite*{idsep} for details.
\begin{example}[Additive Random Utility Models (\cite*{mcfadden1981})]
Let
\[
v_j = u_j + \eps_j
\]
denote the latent utility for alternative $j$. Treat $\eps = (\eps_1, \ldots, \eps_K)$ as a random variable and $u$ as a constant. Normalize the latent utility of the outside good ($j = 0$) to $0$ and assume there are $K$ inside goods. Suppose the individual chooses an alternative that maximizes latent utility and let $\tilde{D}(u, \eps) \in \{0, 1\}^{K}$ be a vector of indicators denoting denoting which, if any, of the inside goods ($j > 0$) is chosen. Then by similar arguments as in the quasilinear utility example, necessary conditions for optimality imply
\[
\left( \tilde{D}(u, \eps) - \tilde{D}(\tilde{u}, \eps) \right) \cdot ( u - \tilde{u} ) \geq 0.
\]
Moreover, letting $Q(u) = \E \left[ \tilde{D}(u, \eps) \right]$ where the expectation is over $\eps$, we have
\[
( Q(u) - Q(\tilde{u}) ) \cdot ( u- \tilde{u} ) \geq 0.
\]
\end{example}
In this example, $Q(u)$ is the vector of probabilities for choosing each of the $K$ inside goods.\footnote{In this paper I treat $u$ as a fixed parameter. If $u$ is treated as an observable random variable, then as long as $u$ is independent of $\eps$ (and some technical conditions are met), $Q(u)$ is the conditional probability of choosing each alternative, conditional on the shifters $u$.} This example illustrates two principles. First, the law of demand is preserved under expectations. In particular, the law of demand holding at the individual level implies it holds at the aggregate level.\footnote{See \cite{shi2018estimating} for a related application of this principle for discrete choice panel data.} Second, injectivity results may be used for aggregate data even when injectivity fails at the individual level. Note that for fixed $\eps$, the function $\tilde{D}(\cdot,\eps)$ cannot be injective whenever $U$ has more than $K+1$ points. However, taking expectations can serve to smooth out discreteness and restore injectivity. Whether $Q$ is injective depends on the distribution of $\eps$ (\cite*{norets2013surjectivity}; see also \cite*{azevedo2013walrasian}).

\section{Law of Demand with a Change of Variables} \label{sec:change}

In some models the law of demand does not hold, but holds after a change of variables. We can adapt  Proposition~\ref{cor:three} to such settings when the change of variables is sufficiently well-behaved. I formalize that the law of demand holds after a change of variables as follows.
\begin{assm} \label{assm:change}
$\tilde{Q}(u) = Q(f(u))$, where $Q : U \subseteq \mathbb{R}^K \rightarrow \mathbb{R}^K$ satisfies the law of demand and is continuous, $f : T \subseteq \mathbb{R}^K \rightarrow U \subseteq \mathbb{R}^K$, and $U$ and $T$ are open and convex.
\end{assm}
When $f$ is a \textbf{homeomorphism}, i.e. a continuous function with a continuous inverse, we still obtain that local injectivity implies global injectivity.
\begin{prop} \label{prop:homeo}
Let Assumption~\ref{assm:change} hold with $f$ a homeomorphism. Then the following are equivalent:
\begin{enumerate}[i.]
    \item $\tilde{Q}$ is injective.
    \item $\tilde{Q}$ is locally injective.
\end{enumerate}
\end{prop}

\begin{proof}
Clearly (i) $\implies$ (ii), and so we wish to show (ii) $\implies$ (i). By Lemma~\ref{lem:invconv}, the set $Q^{-1}(Q(f(u)))$ is convex, hence connected, for each $u \in U$. Since $f^{-1}$ is continuous, its image of the connected set $Q^{-1}(Q(f(u)))$ is connected.\footnote{Recall a set is connected if it cannot be partitioned into two disjoint nonempty sets that are open in the relative topology. Convex sets are clearly connected.} Hence, $\tilde{Q}^{-1}(\tilde{Q}(u)) = f^{-1}(Q^{-1}(Q(f(u))))$ is connected.

By the assumption of local injectivity of $\tilde{Q}$, the set $\tilde{Q}^{-1}(\tilde{Q}(u))$ consists of isolated points. That is, each $\tilde{u} \in \tilde{Q}^{-1}(\tilde{Q}(u))$ has a neighborhood $H$ such that $H \cap \tilde{Q}^{-1}(\tilde{Q}(u)) = \tilde{u}$. Since $\tilde{Q}^{-1}(\tilde{Q}(u))$ is connected, nonempty and consists of isolated points, it can have exactly one point. Since this is true for arbitrary $u$, we obtain part (i).
\end{proof}
An example of a homeomorphism is $f(u) = (f_1 (u_1), \ldots, f_K(u_K))$, where each $f_k$ is strictly increasing and continuous and $T$ is rectangular (i.e. the Cartesian product of intervals). Note that we no longer conclude that checking local injectivity of $\tilde{Q}$ is equivalent to checking whether it is constant on line segments.

We can obtain a sharper result with alternative structure on $f$. Say that a mapping $f : \mathbb{R}^K \rightarrow \mathbb{R}^K$ is \textbf{affine} if it may be written $f(u) = Au + b$ for some $K \times K$ matrix $A$ and vector $b \in \mathbb{R}^K$. Affine mappings need not satisfy the law of demand.\footnote{They do precisely when the symmetrized matrix $(A + A')/2$ is positive semi-definite (e.g. \cite{rock70}, p. 240), where $A'$ denotes the transpose of $A$. This can be seen by writing $(f(u) - f(\tilde{u}))'(u - \tilde{u}) = (u - \tilde{u})'A' (u - \tilde{u})$ and recalling the definition of the law of demand.} Nonetheless, affine mappings have the important property that for a convex set $B \subseteq \mathbb{R}^K$, $f^{-1}(B)$ is convex. By leveraging Lemma~\ref{lem:invconv}, this leads to the following result.

\begin{prop}
Let Assumption~\ref{assm:change} hold for affine $f$. Then for each $y \in \mathbb{R}^K$,
\[
\tilde{Q}^{-1}(y) = \{ u \in T \mid \tilde{Q}(u) = y \}
\]
is convex. In particular, the following are equivalent:
\begin{enumerate}[i.]
    \item $\tilde{Q}$ is injective.
    \item $\tilde{Q}$ is locally injective.
    \item The only line segments in $U$ along which $Q$ is constant are points.
\end{enumerate}
\end{prop}
\begin{proof}
The set $Q^{-1}(\tilde{Q}(u))$ is convex. Thus the set $\tilde{Q}^{-1}(\tilde{Q}(u)) = f^{-1}(Q^{-1}(\tilde{Q}(u)))$ is convex because $f$ is affine. The result is then analogous to Proposition~\ref{cor:three}.
\end{proof}

\section{Discussion}

This paper leverages a classical result in monotone operator theory to provide simple necessary and sufficient conditions to check when a demand mapping is injective. Specifically, for continuous demand mappings that satisfy the law of demand and that are defined over an open convex domain, local injectivity and global injectivity are equivalent. In addition, injectivity can be checked by seeing if the demand mapping is constant over any line segments that are not points. I describe the relationship to a classical result of \cite*{gale1965jacobian} for quasi-definite Jacobians, providing necessary and sufficient conditions for global injectivity in terms of directional derivatives. Finally, I show that the law of demand is not nested in and does not nest the ``connected substitutes'' condition of \cite{berry2013connected}.

\begin{appendices}

\section{Inverse Isotonicity and Substitution} \label{app:pfunction}

To keep the paper self-contained, I provide a primitive condition that ensures the demand mapping satisfies inverse isotonicity.
\begin{lemma}[\cite{more1973p}]
Let $Q : U \subseteq \mathbb{R}^K \rightarrow \mathbb{R}^K$, where $U$ is a Cartesian product. In addition, assume $Q$ satisfies strict own-good monotonicity and weak substitutability. The following are equivalent:
\begin{enumerate}
    \item $Q$ satisfies inverse isotonicity.
    \item $Q$ is a $P$-function, i.e. for $u \neq \tilde{u}$, there is some $k$ such that
    \[
    (Q_k(u) - Q_k(\tilde{u}))(u_k - \tilde{u}_k) > 0.
    \]
\end{enumerate}
\end{lemma}
We note that while \cite{more1973p} prove this result for $U$ a rectangle (i.e. a Cartesian product of intervals), their proofs go through without modification when $U$ is an arbitrary Cartesian product. $P$-functions are closely related to functions whose Jacobians are $P$-matrices, whose injectivity properties are studied in \cite{gale1965jacobian}. See \cite{more1973p} for more details.

From this result we conclude that because the connected substitutes assumption of \cite*{berry2013connected} implies inverse isotonicity, the demand mapping in their setup is a $P$-function. This can be deduced from their Lemma 3, which states that under their assumptions, if $u \neq \tilde{u}$ and $\mathcal{I} = \{ k \mid u_k > \tilde{u}_k \}$ is nonempty, then
\[
\sum_{k \in \mathcal{I}} Q_k(u) > \sum_{k \in \mathcal{I}} Q_k(\tilde{u}).
\]
This implies that there must be some $k \in \mathcal{I}$ such that $(Q_k(\tilde{u}) - Q_k(u))(u_k - \tilde{u}_k) > 0$, i.e. $Q$ must be a $P$-function.\footnote{Note that if $\mathcal{I}$ is empty we can repeat the argument with $u$ and $\tilde{u}$ interchanged.}
\end{appendices}

\bibliographystyle{plainnat}
\bibliography{ref}

\end{document}